\documentclass[conference,letterpaper]{IEEEtran}

\usepackage[T1]{fontenc}

\IEEEoverridecommandlockouts

\makeatletter
\newcommand\HUGE{\@setfontsize\Huge{40}{50}}
\makeatother    

\usepackage{mathrsfs}
\usepackage{amsfonts}
\usepackage{amsmath}
\usepackage{amssymb}
\usepackage{array}
\usepackage{multirow}
\usepackage{multicol}
\usepackage[dvipsnames]{xcolor}
\usepackage{graphicx}
\usepackage[caption=false]{subfig}
\usepackage{tikz}
\usepackage{pgfplots}
\usepackage{comment}
\usepackage{stackengine}
\usepackage{pifont}
\usepackage{booktabs,multirow}
\usepackage{mathtools}
\usepackage{thmtools}
\usepackage{ntheorem}
\usepackage[acronym, shortcuts]{glossaries}
\usepackage[textsize=tiny]{todonotes}
\usepackage[
  top=0.78in,
  bottom=0.78in,
  left=0.73in,
  right=0.73in
]{geometry}

\usepgfplotslibrary{fillbetween}

\newcommand{\6}{\mathbf }

\usepackage[ruled, vlined, linesnumbered]{algorithm2e}

\SetCommentSty{mycommfont}

\usepackage[backend=bibtex,bibstyle=ieee,citestyle=numeric-comp,hyperref=true]{biblatex}
\newcolumntype{P}[1]{>{\centering\arraybackslash}p{#1}} 

\newcommand{\supp}{\mathrm{Supp}}
\newcommand{\wt}{\mathrm{wt}}

\usepackage{mathabx}
\usepackage{pifont}

\usepackage{booktabs}

\usepackage{array}

\newcommand{\code}{\mathscr C}

\newcommand{\nctable}{\mathtt{NC}\_\mathtt{Syndromes}}

\usetikzlibrary{fit,calc}
\newcommand{\boxit}[2]{
    \tikz[remember picture,overlay] \node (A) {};\ignorespaces
    \tikz[remember picture,overlay]{\node[yshift=3pt,fill=#1,opacity=.25,fit={($(A)+(0,0.15\baselineskip)$)($(A)+(.9\linewidth,-{#2}\baselineskip - 0.25\baselineskip)$)}] {};}\ignorespaces
}

\newcommand{\bfmax}{\textsf{BF-Max}}
\newcommand{\mld}{\textsf{MLD}}
\newcommand{\bfout}{\textsf{OoP}}
\newcommand{\bgf}{\textsf{BGF}}
\newcommand{\newbike}{\textsf{BIKE}\text{-}\textsf{flip}}

\newcommand{\PreserveBackslash}[1]{\let\temp=\\#1\let\\=\temp}
\newcolumntype{C}[1]{>{\PreserveBackslash\centering}p{#1}}
\newcolumntype{R}[1]{>{\PreserveBackslash\raggedleft}p{#1}}
\newcolumntype{L}[1]{>{\PreserveBackslash\raggedright}p{#1}}

\usepackage{arydshln}

\newcommand{\nonl}{\renewcommand{\nl}{\let\nl\oldnl}}

{\theorembodyfont{\rmfamily}
   }
{\theorembodyfont{\rmfamily}
   }
{\theorembodyfont{\rmfamily}
   }
{\theorembodyfont{\rmfamily}
   }
{\theorembodyfont{\rmfamily}
   }
{\theorembodyfont{\upshape}
   }
{\theorembodyfont{\upshape}
   }
{\theorembodyfont{\rmfamily}
   }

\newtheorem{proposition}{Proposition}

\newtheorem{definition}{Definition}

\newtheorem{lemma}{Lemma}

\newtheorem{remark}{Remark}

\newtheorem{example}{Example}

\newacronym{LDPC}{LDPC}{low-density parity-check}
\newacronym{MDPC}{MDPC}{moderate-density parity-check}
\newacronym{QC-MDPC}{QC-MDPC}{quasi-cyclic moderate-density parity-check}
\newacronym{KEM}{KEM}{key encapsulation mechanism}
\newacronym{BF}{BF}{Bit Flipping}
\newacronym{DFR}{DFR}{Decoding Failure Rate}
\newacronym{IND-CCA2}{IND-CCA2}{INDistinguishability under Adaptively Chosen Ciphertext Attacks}
\newacronym{IND-CPA}{IND-CPA}{INDistinguishability under Chosen Plaintext Attacks}
\newacronym{QC}{QC}{quasi-cyclic}

\pgfplotsset{compat=1.17}

\title{Near-Codewords Aware Bit Flipping Decoding of QC-MDPC Codes\thanks{The work of Alessio Baldelli was partially supported by Agenzia per la Cybersicurezza Nazionale (ACN) under the programme for promotion of XL cycle PhD research in cybersecurity (CUP I32B24001750005). 
The work of Marco Baldi and Paolo Santini was partially supported by the Italian Ministry of University and Research (MUR) under the Italian Fund for Applied Science (FISA) 2022, Call for tender No. 1405 published on 13/09/2022 - project title “Quantum-safe cryptographic tools for the protection of national data and information technology assets” (QSAFEIT) - No. FISA 2022-00618 (CUP I33C24000520001), Grant Assignment Decree no. 15461 adopted on 02/08/2024.
}}

\author{
\IEEEauthorblockN{
\begin{tabular}{c}
Alessio Baldelli, Marco Baldi, Davide De Zuane, and Paolo Santini  \\
Department of Information Engineering,
Università Politecnica delle Marche, Ancona, Italy \\
\texttt{a.baldelli@pm.univpm.it}, \texttt{davide.dezuane@imtlucca.it}, \\
\{\texttt{m.baldi}, \texttt{p.santini}\}\texttt{@univpm.it}. 
\end{tabular}
}
}

\begin{document}

\maketitle

\begin{abstract}
Bit-Flipping (BF) decoders are a family of decoders widely employed in post-quantum cryptographic schemes based on Quasi-Cyclic Moderate-Density Parity-Check (QC-MDPC) codes, such as BIKE. BF decoders suffer from trapping sets, corresponding to low-weight error patterns that likely lead to decoding failures. For QC-MDPC codes, the most relevant family of trapping sets is that of \emph{near-codewords}, which are error patterns associated to low-weight syndromes. Indeed, recent works show that error patterns having a large overlap with near-codewords are the main culprits for decoding failures at very low Decoding Failure Rate (DFR) values. In this paper, we show that any BF decoder can be tweaked and made somehow aware of near-codewords, which means being able to recognize, and recover from, bad configurations due to near-codewords.  We show that this modification results in minimal computational overhead. Through intensive numerical simulations, we evaluate the effectiveness of this approach on several BF decoders, considering both toy code parameters and BIKE parameters for NIST security category 1. Our results show drastic reductions in the DFR. We also find that, with this modification, a recently proposed BF variant called $\bfmax$ outperforms the two decoders used by BIKE within the NIST competition.
\end{abstract}

\section{Introduction}
\label{sec:intro}

One of the most promising ways to build post-quantum Key Encapsulation Mechanisms from codes is to rely on the McEliece/Niederreiter framework instantiated with Quasi Cyclic Moderate-Density Parity-Check (QC-MDPC) codes \cite{misoczki2013mdpc, aragon2020bike, leda}.
In such schemes, decapsulation boils down to decoding a low-weight error through the underlying QC-MDPC code.
Among the available decoders, the Bit Flipping (BF) algorithm is normally the preferred choice, because of its excellent trade-off between simplicity, computational efficiency and error correction capability
\cite{drucker2020qc, sendrier2019decoding, vasseur2021thesis, leda, aragon2020bike}. 
As a matter of fact, BIKE (one of the finalists in the NIST competition for the standardization of post-quantum cryptography 
\cite{alagic2025status}) is based on QC-MDPC codes and relies on the BF decoder. 

However, BF decoders are characterized by a non-zero Decoding Failure Rate (DFR). 
This is an issue, since \emph{decoding failures} lead to \emph{decryption failures} that, in turn, are associated with some information leakage about the secret key.
This leads to practical attacks \cite{guo2016key} and ultimately hinders achieving INDistinguishability under Adaptively Chosen Ciphertext Attacks (IND-CCA2).
To achieve IND-CCA2, the DFR of the employed decoder must be \emph{negligible} in the security parameter: for instance, for $128$ bits of security, we need $\text{\ac{DFR}} < 2^{-128}$ \cite{hofheinz2017modular}. 
For this reason, studying and modeling the DFR of QC-MDPC codes at very low values has become highly relevant \cite{annechini2024bitflipping, baldelli2025bf, arpin2025error, arpin2022study, vasseur2021thesis, sendrier2020low, tillich2018decoding}.

 
For the 
QC-MDPC matrices used in BIKE and other cryptographic schemes, consisting of two circulant blocks, some works have identified a family of trapping sets, called \textit{near-codewords}, which are formed by a column of one of the circulant blocks themselves \cite{baldi2021performance, arpin2022study, vasseur2021thesis}. 
Indeed, whenever the error pattern has a large overlap with a near-codeword, the decoder is more likely to get fooled and to converge to the near-codeword.
This situation is analyzed accurately in  \cite{arpin2025error}; the results in the paper imply that, at very low DFR values, near-codewords are the most meaningful events of failure.


\paragraph*{Our contribution}
In this paper, we show how the knowledge of near-codewords for one QC-MDPC code can be employed to improve BF decoding of the same code.
The intuition is the following: when the initial error pattern overlaps with a near-codeword in many positions (we refer to such errors as \textit{almost near-codewords}), the decoder is likely to get stuck in a bad configuration in which the errors to be corrected correspond exactly to the near-codeword.
Since near-codewords are known, their syndromes are known as well: by observing the syndrome, the bad configuration can be recognized and the residual errors can be corrected accordingly. 
This can be done efficiently by reading from a look-up table where syndromes of near-codewords are stored. 
The cost of accessing the table can be made very small, and much less than the cost of running the BF decoder, using proper data structures (e.g., sorting the syndromes and accessing the list via binary searches).
In the end, this modification does not lead to any relevant computational overhead.
Moreover, this improvement can be applied to \emph{any} BF decoder as it is entirely agnostic to what the decoder does precisely (e.g., in-place vs out-of-place flipping strategy, how many iterations are performed, how many bits are flipped in each iteration, how bits to be flipped are selected, etc.).

To validate our approach, we consider several BF decoders from the literature and show that, for each of them, the DFR is always reduced.
First, we focus on codes with toy parameters, which allow us to verify the impact on the DFR in the error floor region.
Then, we consider the actual BIKE code parameters for NIST security category 1.
In this case, we are obviously not able to simulate the error floor, but we verify the validity of our idea by simulating those errors having large intersections with near-codewords.
Moreover, our results show that \bfmax, a recently proposed step-by-step decoder \cite{baldelli2025bf},  outperforms all the other decoders regarding decoding of almost near-codewords.
Remarkably, in all our simulations, the modified $\bfmax$ has never encountered a decoding failure.

\section{Notation and Background}
\label{sec:notation}

We denote with $\mathbb{F}_2$ the binary field.
Vectors and matrices are indicated with bold lowercase and uppercase letters, respectively.
For a vector $\6a$, $a_i$ indicates its $i$-th entry; analogously, for a matrix $\6A$, $a_{i,j}$ indicates the entry in the $i$-th row and $j$-th column.
The all-zero vector of length $r$ is denoted by $\60_r$.
The support of a vector $\6a$, which is the set of indexes of its non-zero entries, is $\supp(\6a)$, while $\mathrm{wt}(\6{a})$ denotes its \emph{Hamming weight}.
For a vector $\6a\in\mathbb F_2^n$ and a set $J\subseteq\{0,\cdots,n-1\}$, $\6a|_J$ stands for the vector formed by the entries of $\6a$ which are indexed by $J$.
A square matrix is said to be \textit{circulant} if every row can be obtained by applying a cyclic shift to the previous row; the same holds for the columns.
As a consequence, in a circulant matrix, all rows and columns have the same Hamming weight; we will sometimes refer to it as \textit{circulant weight}.
The set of $r \times r$ binary circulant matrices equipped with the standard matrix sum and multiplication forms a ring which is isomorphic to $\mathbb F_2[x]/(x^r+1)$.
Namely, let $(a_0,\cdots,a_{r-1})$ be the first row of a circulant matrix: the associated polynomial is $\sum_{i = 0}^{r-1}a_ix^i \in \mathbb F_2[x]/(x^r+1)$.
We will be rather flexible with the notation and switch between the matrix and polynomial notations from time to time.

\subsection{Linear codes and QC-MDPC codes}

A binary \emph{linear-code} $\code \subseteq \mathbb F_2^n$ with length $n$ and redundancy $r \leq n$ is a linear subspace of $\mathbb F_2^n$ with dimension $k = n - r$.
Any linear code can be represented with a \emph{parity-check matrix}, that is, an $r\times n$ full-rank matrix such that 
$\code \coloneqq \left\{\6c\in\mathbb F_2^n\mid \6c\6H^\top = \60\right\}.$
Given $\6e \in \mathbb{F}_2^n$, its associated \textit{syndrome} is $\6s = \6e\6H^\top \in \mathbb{F}_2^r$.
If $\6s \neq \60_r$, then $\6e \not \in \code$ otherwise, if $\6s = \60_r$, then $\6e \in \code$.
\emph{\Ac{LDPC}} codes are a family of error-correcting codes represented by \emph{sparse} parity-check matrices, while
\ac{MDPC} codes are \ac{LDPC} codes with somewhat less sparse parity-check matrices\footnote{Some authors actually distinguish \ac{LDPC} codes from \ac{MDPC} codes by saying that parity-check equations for \ac{LDPC} codes have weight which grows as $O\big(\log(n)\big)$, while for \ac{MDPC} codes it grows as $O\big(\sqrt{n}\big)$ \cite{ouzan2009, misoczki2013mdpc,chiaraluce_2007}.}.
\ac{MDPC} and \ac{LDPC} codes share many structural properties and can be decoded with the same algorithms.
A code is said to be QC if it is invariant under block-wise cyclic shifts.
We introduce next the \ac{QC} codes we consider in this paper.
\begin{definition}\label{def:qc_mdpc}
Let $\6H = (\6H_1, \6H_2)$, with $\6H_1$ and $\6H_2$ being $r \times r$ binary circulant matrices. 
Let $\code$ be the code having $\6H$ as parity-check matrix. 
Then, we say that $\code$ is a \emph{double circulant code}.
\end{definition}
When both $\6H_1$ and $\6H_2$ have weight $v\ll r$, then we speak of \textit{double circulant QC-MDPC codes}.
We will sometimes write $\6H = (h_1, h_2)$, where $h_1, h_2\in\mathbb F_2[x]/(x^r+1)$ are the polynomials associated with $\6H_1^\top$ and $\6H_2^\top$, respectively.
This becomes rather handy when expressing the syndrome as a polynomial: for a certain $\6e = (\6e_1, \6e_2)\in\mathbb F_2^{2r}$, if we denote by $e_1$ and $e_2$ the polynomials associated to $\6e_1$ and $\6e_2$, respectively, we can associate the syndrome $\6s = \6e\6H^\top$ to the polynomial
$s = e_1\cdot h_1 + e_2 \cdot h_2$.

\subsection{Bit Flipping decoding}

The simple idea that underlies every BF decoder is to iteratively flip the bits that are more likely to be in error, based on the so-called \textit{counters}.
\begin{definition}
Let $\6H\in\mathbb F_2^{r\times n}$ be a parity-check matrix and $\6s\in\mathbb F_2^r$ a 
syndrome.
Then, for $i\in\{0,\cdots,n-1\}$, the $i$-th \emph{counter} is defined as
$$\sigma_i = \big| \left\{\left.j\in\{0,\cdots,r-1 \right\} \, | \, (h_{j,i} = 1)\wedge (s_j = 1)\right\} \big|.$$
Equivalently, let $\6h_i$ indicate the $i$-th column of $\6H$; then, the $i$-th counter is defined as
$\sigma_i = \wt\big(\6s\mid_{\supp(\6h_i)}\big)$.
\end{definition}
Counters can be used to locate errors: if $\sigma_i$ is large, then it is likely that $e_i = 1$.
This is the basis of any BF decoder.
In Algorithm \ref{alg:bf_general}, we sketch the blueprint of a BF decoder.
As we can see, counters are crucial to determine error positions.
Later on, we will describe which set $J$ of large counters is considered in several BF variants.

\begin{algorithm}[t]
\small{\KwData{parity-check matrix $\6H\in\mathbb F_2^{r\times n}$, maximum number of iterations $\mathtt{IterMax}\in\mathbb N$}
\KwIn{syndrome $\6s \in \mathbb F_2^r$}
\KwOut{decoding failure $\bot$, or vector $\widehat{\6e}\in\mathbb F_2^n$ if $\6s = \6H\widehat{\6e}^\top$}

\tcc{Initialize error estimate and iterations}
Set $\widehat{\6e} = \6{0}_n$, \quad $\mathtt{Iter} = 0$\;
\While{$\big(\6s\neq \60\big)\wedge\big(\mathtt{Iter}<\mathtt{IterMax}\big)$}{
Compute counters $(\sigma_0, \cdots, \sigma_{n-1})$\;

\tcc{Find positions that are likely to be wrongly estimated}
Determine $J\subseteq\{0,\cdots,n-1\}$ so that, for each $i\in J$, $\sigma_i$ is large enough\;

\tcc{Update error estimate and syndrome}
\For{$i\in J$}{
$\widehat e_{i}\gets \widehat e_{i} + 1$, \quad 
$\6s\gets \6s + \6h_{i}$\;
}
\tcc{Update the number of iterations}
$\mathtt{Iter}\gets \mathtt{Iter}+1$\;
}

\tcc{Return failure if syndrome is not null}
\textbf{if} $\6s\neq \60$ return $\bot$, \textbf{else} return $\widehat{\6e}$\;

\caption{\textsf{BF}\label{alg:bf_general}
}
}
\end{algorithm}

We finally recall a useful approximation regarding counters. 
This is related to the elements of the adjacency matrix $\6A\in\mathbb N^{n\times n}$, introduced in \cite{santini_hard}, which is defined as follows
$$
a_{i,j} = \begin{cases}
0 &\text{if $i = j$},\\
\left|\supp(\6h_i)\cap \supp(\6h_j)\right| &\text{otherwise.}
\end{cases}
$$
As shown in \cite{santini_hard}, due to the sparsity of $\6H$, for a random error vector $\6e\in\mathbb F_2^n$ with low weight we have
\begin{equation}
\label{eq:counters_approx}
\sigma_i \approx \begin{cases}
v - \sum_{j\in\supp(\6e)\setminus\{i\}}a_{i,j} & \text{if $i\in \supp(\6e)$,}\\
\sum_{j\in\supp(\6e)}a_{i,j} & \text{if $i\not\in \supp(\6e)$.}
\end{cases}
\end{equation}

\subsection{Near-codewords}

We first recall the definition of near-codewords \cite{vasseur2021thesis}.
\begin{definition}
Let $\6H = (h_1, h_2)$ be the parity-check matrix of a double circulant code and $J = \{ 0, \dots, r-1 \}$, then the set of near-codewords is 
$\mathcal{E} = \left\{(x^i\cdot h_1, 0)\mid i \in J \right\} \, \, \cup
\left\{(0,x^i\cdot h_2)\mid i \in J \right\}.$
\end{definition}
Observe that, for every element in $\mathcal E$, the syndrome has weight $v$, since $\mathrm{wt}(h_i^2)=v$.
We generalize the set of near-codewords by introducing \textit{almost near-codewords}, i.e., errors characterized by some intersection with near-codewords.
As shown in \cite{arpin2025error}, when such an intersection is large enough, this type of error patterns become the most meaningful events of failure at very low DFR values.
We formalize this family of error patterns using the parameters $t, u \in \mathbb{N}$, as follows.

\begin{definition} \label{def:almost}
For a parity-check matrix $\6H$ for which $\mathcal{E}$ is the set of near-codewords, let $\mathcal{E}_{t,u} \subseteq \mathbb{F}_2^n$ be the set of $t$-weight vectors such that, for each $\6e\in\mathcal E_{t,u}$, it holds that
$\max_{\6m\in\mathcal E}\left\{|\supp(\6e)\cap \supp(\6m)|\right\} = u$.
We refer to elements of $\mathcal E_{t,u}$ as $(t,u)$\emph{-almost near-codewords}.
\end{definition}

\begin{remark}
    For a $(t,u)$-almost near-codeword $\6e$ there exists at least a near-codeword $\6m \in \mathcal E$ which overlaps with $\6e$ in \emph{exactly} $u$ ones; for all the other near-codewords, the overlapping ones with $\6e$ are \emph{at most} $u$.
\end{remark}
Note that, for a double-circulant code, every error vector of weight $t$ is a $(t,u)$-almost near-codeword for some $u\geq 1$.
\begin{definition}
Let $\6e$ be a $(t,u)$-almost near-codeword and $\6m\in\mathcal E$ such that $|\supp(\6e)\cap \supp(\6m)| = u$. Then, we say that $\6m$ is a \emph{close near-codeword} for $\6e$.
\end{definition}
We recall the notion of \textit{bad bits} and \textit{suspicious bits} from \cite{arpin2025error}. 
\begin{definition} 
Let $\6e\in \mathcal E_{t,u}$ be a $(t,u)$-almost near-codeword. 
Then, we say that 
\begin{itemize}
\item the $j$-th bit is a \emph{bad bit} if $j\in \supp(\6e)$ and there is a close near-codeword $\6m$ for $\6e$ such that $j\in\supp(\6m)\cap\supp(\6e)$;
\item the $j$-th bit is a \emph{suspicious bit} if $j\not\in\supp(\6e)$ and there is a close near-codeword $\6m$ for $\6e$ such that $j\in\supp(\6m)\setminus \supp(\6e)$.
\end{itemize}
\end{definition}
When for a $(t,u)$-almost near-codeword $\6e$ there is only \emph{one} close near-codeword, the corresponding number of bad and suspicious bits is exactly $u$ and $v-u$, respectively. 
If there are more close near-codewords, then these amounts may be larger.
For large $u$, however, there is a unique close near-codeword for each $(t, u)$-almost near-codeword with large probability.

\section{Correcting near codewords by matching syndromes}
\label{sec:lookup}

We first provide a simple explanation of why $(t,u)$-almost near-codewords with a large $u$ are likely to make the BF decoder converge to their close near-codewords\footnote{Notice that this phenomenon is already explained in \cite{arpin2025error}; however, we give a much simpler justification, using the approximation in \eqref{eq:counters_approx}.
This is helpful for understanding our modification to BF decoders}.
Then, we describe our tweak for BF decoders and we evaluate the benefit it brings in terms of performance.

\subsection{Counters for almost near-codewords}

Let us consider a double circulant QC-MDPC code and decoding of a $(t,u)$-almost near-codeword $\6e$. 
We first recall the following useful lemma from \cite{baldi2021performance}.
\begin{lemma} \label{lem:circ}
Let $\6H\in\mathbb F_2^{r\times 2r}$ be a double circulant matrix.
Let $i\neq j$, $i,j\in\mathrm{Supp}(\6h_1)$: then, $a_{i,j}\geq 1$, i.e., columns $i$ and $j$ intersect in at least one position.
Analogously, if $i\neq j$, $i,j\in\mathrm{Supp}(\6h_2)$, then $a_{r+i,r+j}\geq 1$, i.e., columns $r+i$ and $r+j$ intersect in at least one position.
\end{lemma}

\paragraph*{Almost near-codewords with $t = u$}
When $t = u$, all the error bits in $\6e$ are actually bad bits.
Recalling \eqref{eq:counters_approx}, each counter for a bad bit $i$ is approximately $v-\sum_{j\in\supp(\6e)\setminus \{i\}}a_{i,j}$.
Thanks to the sparsity of $\6H$, the majority of the entries $a_{i,j}$ is actually $1$, hence the counter associated to each bad bit is expected to be close to $v-(t-1) = v-u+1$.
For the same reasons, it is possible to approximate the counter value for each suspicious bit as $\sum_{j\in\supp(\6e)}a_{i,j}\approx t = u$.
In Fig.~\ref{fig:counters_distrib}, we show a numerical validation of this approximation.
When $u$ is large enough, say, $u>v/2$, we have that counters for bad bits are expected to be below $v/2$ and counters for suspicious bits exceed $v/2$.
Thus, any BF decoder is going to get fooled and, with high probability, will actually choose the bits to be flipped among the suspicious bits.
This explains the convergence to a close near-codeword.

\begin{figure}[t]
\centering

\begin{minipage}{0.48\linewidth}
    \centering
    \subfloat[$r = 2003$, $t = u = 3$]{%
        \resizebox{\linewidth}{!}{%
            \begin{tikzpicture}
\begin{axis}[
  width=6cm,
  height=4cm,
  ymin=0, ymax=1,
  xmin=0, xmax=15,
  grid = both,
  xtick = {0,3,6,9,12, 15},
  ytick = {0, 0.25, 0.5, 0.75, 1},
  tick label style={font=\normalsize},
  legend columns=-1,
  legend entries={aaa, bbb},
  legend to name=named,
]

\addplot[ybar,bar width=5pt,fill=red!60,opacity=0.6]
table[col sep=comma,x index=0,y index=1] {Data/pr_counters_p2003_v15_u3_t3.txt};
  
\addplot[ybar,bar width=5pt,fill=blue!80,opacity=0.6]
table[col sep=comma,x index=0,y index=2] {Data/pr_counters_p2003_v15_u3_t3.txt};

\addplot[red!60, densely dashdotted, opacity = 0.9, thick]coordinates
{(3, 0)
(3, 1)};

\addplot[blue!80, densely dashdotted, opacity = 0.9, thick]coordinates
{(13, 0)
(13, 1)};

\node[red!60, opacity = 0.9, thick](u) at (axis cs: 2, 1.1)   {};



\coordinate (A) at (axis cs: 3, 1);%
\coordinate (B) at (axis cs: 13, 1);%

\end{axis}

\node[red!60, opacity = 0.9, above](pu) at (A)   {$u$};

\node[blue!80, opacity = 0.9, above](pu) at (B)   {$v-u+1$};

\end{tikzpicture}%
        }%
    }
    \label{fig:counters_3}
\end{minipage}
\hfill
\begin{minipage}{0.48\linewidth}
    \centering
    
\subfloat[$r=2003$, $t=u=13$]{%
        \resizebox{\linewidth}{!}{%
            \begin{tikzpicture}
\begin{axis}[
  width=6cm,
  height=4cm,
  ymin=0, ymax=1,
  xmin=0, xmax=15,
  grid = both,
  xtick = {0,3,6,9,12, 15},
  ytick = {0, 0.25, 0.5, 0.75, 1},
  tick label style={font=\normalsize},
  legend columns=-1,
  legend entries={aaa, bbb},
  legend to name=named,
]

\addplot[ybar,bar width=5pt,fill=red!60,opacity=0.6]
table[col sep=comma,x index=0,y index=1] {Data/pr_counters_p2003_v15_u13_t13.txt};
  
\addplot[ybar,bar width=5pt,fill=blue!80,opacity=0.6]
table[col sep=comma,x index=0,y index=2] {Data/pr_counters_p2003_v15_u13_t13.txt};

\addplot[red!60, densely dashdotted, opacity = 0.9, thick]coordinates
{(13, 0)
(13, 1)};

\addplot[blue!80, densely dashdotted, opacity = 0.9, thick]coordinates
{(3, 0)
(3, 1)};

\node[red!60, opacity = 0.9, thick](u) at (axis cs: 13, 1.1)   {};



\coordinate (A) at (axis cs: 13, 1);%
\coordinate (B) at (axis cs: 3, 1);%

\end{axis}

\node[red!60, opacity = 0.9, above](pu) at (A)   {$u$};

\node[blue!80, opacity = 0.9, above](pu) at (B)   {$v-u+1$};

\end{tikzpicture}%
        }%
    \label{fig:counters_13}
}

\end{minipage}

\caption{Experimental distribution of counters for bad bits (in blue) and suspicious bits (in red), for codes with $r = 2003$ and different circulant weights.
For each configuration, we have generated uniformly at random $100$ $(t,u)$-almost near-codewords.}
\label{fig:counters_distrib}
\end{figure}
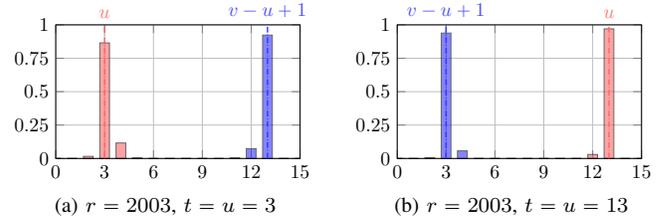



\subsubsection*{Almost near-codewords with $t > u$}

Now, with respect to the previous scenario, we have to deal with a larger number of errors.
The decoder will correct errors that are not bad bits with large probability, since their counters are expected to be large.
Analogously, error-free positions that are not suspicious bits should have low counters, hence the decoder will not flip them, with large probability.
So, again, when $u$ is large enough, with large probability the decoder will converge to a close near-codeword.
Actually, in this case we expect that convergence may start even when $u < v/2$, since the extra errors are going to worsen the counters behavior, overall.

\subsection{Recognizing syndromes of near-codewords}

In Algorithm~\ref{alg:bf_lookup}, we show how a BF decoder can be modified to exploit knowledge of near-codewords.
Observe that the only modification corresponds to instructions 7--10.

To do this, we prepare a look-up table as follows
$$\nctable =  \left\{ \big(\6m_i\6H^\top, i\big) \, | \, i = 0,\cdots,2r-1 \right\},$$
which contains the syndromes of the near-codewords (denoted as $\6m_0, \dots,\6m_{2r-1}$), together with the associated indexes.
In polynomial notation, we have
$$\6m_i\6H^\top = \begin{cases}
x^i\cdot h_1^2 &\text{if $i<r$,}\\
x^{i-r}\cdot h_2^2 &\text{if $i\geq r$.}\\
\end{cases}$$
At the end of each iteration, the decoder checks if the updated syndrome coincides with one of the syndromes stored in the look-up table.
Whenever this happens, the decoder recognizes that it got stuck in a near-codeword and recovers from this bad situation by flipping the bits corresponding to the near-codeword.
\begin{algorithm}[t]
\small{\KwData{parity-check matrix $\6H\in\mathbb F_2^{r\times n}$, maximum number of iterations $\mathtt{IterMax}\in\mathbb N$, table $\nctable$}
\KwIn{syndrome $\6s \in \mathbb F_2^r$}
\KwOut{decoding failure $\bot$, or vector $\widehat{\6e}\in\mathbb F_2^n$ if $\6s = \6H\widehat{\6e}^\top$}
\vspace{1mm}

\tcc{Initialize error estimate and iterations}
Set $\widehat{\6e} = \6{0}_n$, \quad $\mathtt{Iter} = 0$\;
\While{$\big(\6s \neq \60_r \big)\vee\big(\mathtt{Iter}<\mathtt{IterMax}\big)$}{
Compute counters $(\sigma_0, \cdots, \sigma_{n-1})$\;

\tcc{Find positions that are likely to be wrongly estimated}
Determine $J\subseteq\{0, \cdots, n-1\}$ so that, for each $i\in J$, $\sigma_i$ is large enough\;

\tcc{Update error estimate and syndrome}
\For{$i\in J$}{
$\widehat e_{i}\gets \widehat e_{i} + 1$, \quad $\6s\gets \6s + \6h_{i}$\;
}

{\nonl\boxit{Gray!60!White}{4.8}}
\tcc{Scan the syndromes of near-codewords}
\If{$\wt(\6s) = v$}{
\If{$\nctable$ \emph{contains} $\6s$ \emph{at position} $i$} 
{
Update $\widehat{\6e}$ as $\widehat{\6e} + \6m_i$\;
\Return $\widehat{\6e}$\;
}
}

\tcc{Update the number of iterations}
$\mathtt{Iter}\gets \mathtt{Iter}+1$\;

}
\tcc{Return failure if syndrome is not null}
\textbf{if} $\6s \neq \60$ return $\bot$, \textbf{else} return $\widehat{\6e}$\;
\caption{\textsf{BF} with knowledge of near-codewords}
\label{alg:bf_lookup}
}
\end{algorithm}
Since the syndrome of a near-codeword has exactly weight $v$, we access the table only when the weight of the updated syndrome is exactly $v$.
This avoids useless accesses to the look-up table.
Moreover, as we show with the next proposition, the time and the space overheads are very mild.
\begin{proposition}
    Storing the table $\nctable$ requires $O\big( r\cdot v\cdot\log_2(r) \big)$ memory, while each access to the table requires  
    $O\big( \log^2_2(r) \cdot v \big)$ time.
\end{proposition}
\begin{IEEEproof}
Since each syndrome in the table has support size $v$, we can store it using the elements of its support, with only $v\cdot\log_2(r)$ bits.
The syndromes can then be sorted, so that accessing the table can be done with a binary search algorithm.
This, in the worst case, requires $\log_2(2r)$ comparisons between objects having binary length $v\cdot\log_2(r)$.
\end{IEEEproof}
\begin{remark}
QC-MDPC codes used in cryptography have $v = O(\sqrt{r})$, which leads to BF decoding having a time complexity of $O(r^{1.5})$.
This is much more than the overhead due to reading from the look-up table, which is 
$O\big(\sqrt{r}\cdot \log_2^2(r)\big)$.
\end{remark}

\section{Numerical results}
\label{sec:bike}

In this section, we show that our modification effectively leads to significant reductions in the DFR.
To this end, we performed intensive numerical experiments considering several BF decoders.
All code and data used in our simulations have been made publicly available to ensure the reproducibility of our results\footnote{\url{https://github.com/secomms/BF_near_codewords}}.
\vspace{-0.10cm}

\subsection{Small codes}

We first focus on “toy” double-circulant QC-MDPC codes, with $r = 2003$ and $v\in\{ 9, 11, 13, 15 \}$. 
We consider three decoders, all based on the blueprints from  Algorithms \ref{alg:bf_general} and \ref{alg:bf_lookup}. $\mathtt{IterMax}$ denotes the maximum number of iterations.

The first one is the $\bfmax$ \cite{baldelli2025bf} decoder, which flips only one bit in each iteration, choosing the one with the largest counter value.
If there are more counters with the same value, the decoder randomly chooses the bit to be flipped among them.
In Fig.~\ref{fig:bfmax_2003}, we provide the results of the corresponding simulations, using 
$\mathtt{IterMax} = 2t$.
 
The second one is the Majority-Logic Decoder ($\mld$), which considers a fixed threshold in each iteration, equal to $\lceil v/2 \rceil$. Then, all the bits for which the associated counters are greater than $ v/2$ are flipped. 
Fig.~\ref{fig:mld_2003} provides the corresponding results of numerical simulations, using $\mathtt{IterMax} = 50$.

The third decoder is the Out-Of-Place BF ($\bfout$) decoder, which uses a potentially different threshold in each iteration. 
It flips all the bits for which counters are at least as large as the current threshold.
For our simulations, we have set $\mathtt{IterMax} = 15$ and used the threshold values reported in Table \ref{tab:thresholds_r2003}.
These values have been chosen with the following criterion: the first iterations use large thresholds (aiming to guarantee that wrong flips occur with small probability), 
while the subsequent iterations use smaller ones (since, at that stage, we expect that most errors have been corrected, then wrong flips occur with small probability in any case).
This mimics the threshold selection rule of more advanced BF decoders, which are studied next. The corresponding simulation results are presented in Fig.~\ref{fig:bfout_2003}.

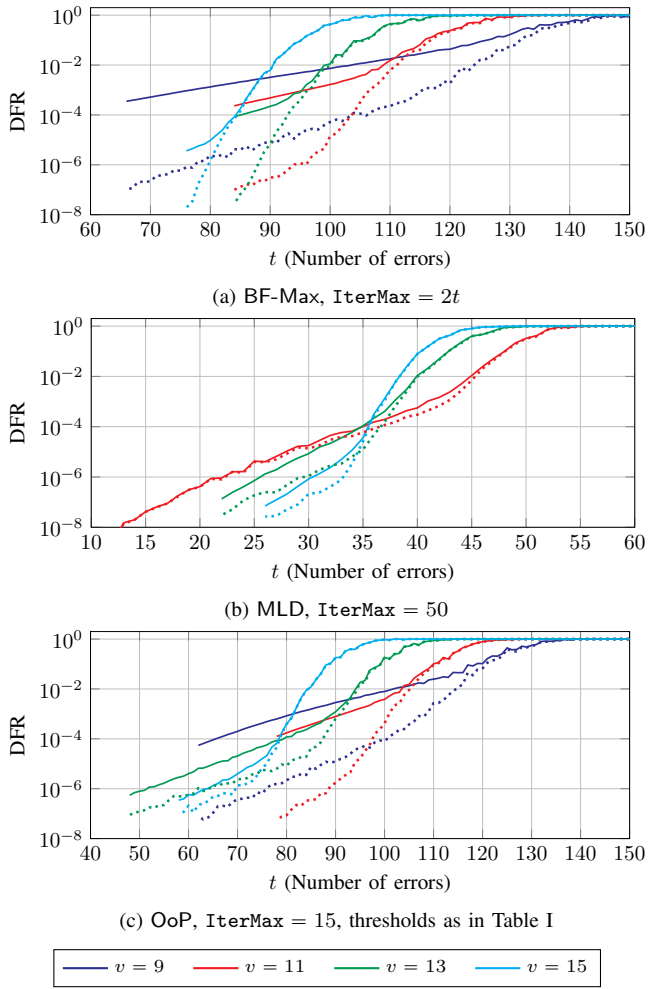
\begin{figure}[h]
    \centering

    \subfloat[$\bfmax$, $\mathtt{IterMax} = 2t$]{%
        \makebox[\linewidth]{%
            \resizebox{\linewidth}{!}{%
                \begin{tikzpicture} 
    \begin{semilogyaxis}[
        xlabel={$t$ (Number of errors)},
        ylabel={ DFR },
        xmin=60, xmax=150,
        ymin=1e-8, ymax=2,
        xtick={60, 70, 80, 90, 100, 110, 120, 130, 140, 150},
        ytick={1, 1e-2, 1e-4, 1e-6, 1e-8},
        legend pos = south east,
        legend style={font=\scriptsize},
        ymajorgrids=true,
        yminorgrids=true,
        xmajorgrids=true,
        height  = 5.0cm,
        minor grid style = {gray!30!white},
        width = 1.2\columnwidth,
        legend entries={$v = 9\hspace{3mm}$, $v = 11\hspace{3mm}$, $v = 13\hspace{3mm}$, $v = 15$},
        legend to name=named,
    ]

\pgfplotstableread{Data/bfmax_2003_9}{\vnine}
\addplot [Blue, thick] table[x=X, y=Y] {\vnine};

\pgfplotstableread{Data/bfmax_2003_11}{\vnine}
\addplot [Red, thick] table[x=X, y=Y] {\vnine};

\pgfplotstableread{Data/bfmax_2003_13}{\vnine}
\addplot [ForestGreen, thick] table[x=X, y=Y] {\vnine};

\pgfplotstableread{Data/bfmax_2003_15}{\vnine}
\addplot [Cyan, thick] table[x=X, y=Y] {\vnine};

\pgfplotstableread{Data/bfmax_2003_9}{\vnine}
\addplot [Blue, dotted, very thick] table[x=X, y=Z] {\vnine};

\pgfplotstableread{Data/bfmax_2003_11}{\vnine}
\addplot [Red, dotted, very thick] table[x=X, y=Z] {\vnine};
     
\pgfplotstableread{Data/bfmax_2003_13}{\vnine}
\addplot [ForestGreen, dotted, very thick] table[x=X, y=Z] {\vnine};
    
\pgfplotstableread{Data/bfmax_2003_15}{\vnine}
\addplot [Cyan, dotted, very thick] table[x=X, y=Z] {\vnine};

\end{semilogyaxis}
\end{tikzpicture}%
            }%
        }%
        \label{fig:bfmax_2003}
    }
    \vspace{+0.00cm}
    \subfloat[$\mld$, $\mathtt{IterMax} = 50$]{%
        \makebox[\linewidth]{%
            \resizebox{\linewidth}{!}{%
                \begin{tikzpicture} 
    \begin{semilogyaxis}[
        xlabel={$t$ (Number of errors)},
        ylabel={ DFR },
        xmin=10, xmax=60,
        ymin=1e-8, ymax=2,
        xtick={5, 10, 15, 20, 25, 30, 35, 40, 45, 50, 55, 60},
        ytick={1, 1e-2, 1e-4, 1e-6, 1e-8},
        legend pos = south east,
        legend style={font=\scriptsize},
        ymajorgrids=true,
        yminorgrids=true,
        xmajorgrids=true,
        height  = 5.0cm,
        minor grid style = {gray!30!white},
        width = 1.2\columnwidth,
        legend columns=-1,
        legend entries={$v = 9\hspace{3mm}$, $v = 11\hspace{3mm}$, $v = 13\hspace{3mm}$, $v = 15$},
        legend to name=named,
    ]

\pgfplotstableread{Data/mld_2003_9}{\vnine}
\addplot [Blue, thick] table[x=X, y=Y] {\vnine};

\pgfplotstableread{Data/mld_2003_11}{\vnine}
\addplot [Red, thick] table[x=X, y=Y] {\vnine};

\pgfplotstableread{Data/mld_2003_13}{\vnine}
\addplot [ForestGreen, thick] table[x=X, y=Y] {\vnine};

\pgfplotstableread{Data/mld_2003_15}{\vnine}
\addplot [Cyan, thick] table[x=X, y=Y] {\vnine};

\pgfplotstableread{Data/mld_2003_9}{\vnine}
\addplot [Blue, dotted, very thick] table[x=X, y=Z] {\vnine};

\pgfplotstableread{Data/mld_2003_11}{\vnine}
\addplot [Red, dotted, very thick] table[x=X, y=Z] {\vnine};
     
\pgfplotstableread{Data/mld_2003_13}{\vnine}
\addplot [ForestGreen, dotted, very thick] table[x=X, y=Z] {\vnine};
    
\pgfplotstableread{Data/mld_2003_15}{\vnine}
\addplot [Cyan, dotted, very thick] table[x=X, y=Z] {\vnine};

\end{semilogyaxis}
\end{tikzpicture}%
            }%
        }%
        \label{fig:mld_2003}
    }
    \vspace{+0.00cm}
    \subfloat[$\bfout$, $\mathtt{IterMax} = 15$, thresholds as in Table~\ref{tab:thresholds_r2003}]{%
        \makebox[\linewidth]{%
            \resizebox{\linewidth}{!}{%
                \begin{tikzpicture} 
    \begin{semilogyaxis}[
        xlabel={$t$ (Number of errors)},
        ylabel={ DFR },
        xmin=40, xmax=150,
        ymin=1e-8, ymax=2,
        xtick={40, 50, 60, 70, 80, 90, 100, 110, 120, 130, 140, 150},
        ytick={1, 1e-2, 1e-4, 1e-6, 1e-8},
        legend pos = south east,
        legend style={font=\scriptsize},
        ymajorgrids=true,
        yminorgrids=true,
        xmajorgrids=true,
        height  = 5.0cm,
        minor grid style = {gray!30!white},
        width = 1.2\columnwidth,
        legend columns=-1,
        legend entries={$v = 9\hspace{3mm}$, $v = 11\hspace{3mm}$, $v = 13\hspace{3mm}$, $v = 15$},
        legend to name=named,
    ]

\pgfplotstableread{Data/out_2003_9}{\vnine}
\addplot [Blue, thick] table[x=X, y=Y] {\vnine};

\pgfplotstableread{Data/out_2003_11}{\vnine}
\addplot [Red, thick] table[x=X, y=Y] {\vnine};

\pgfplotstableread{Data/out_2003_13}{\vnine}
\addplot [ForestGreen, thick] table[x=X, y=Y] {\vnine};

\pgfplotstableread{Data/out_2003_15}{\vnine}
\addplot [Cyan, thick] table[x=X, y=Y] {\vnine};

\pgfplotstableread{Data/out_2003_9}{\vnine}
\addplot [Blue, dotted, very thick] table[x=X, y=Z] {\vnine};

\pgfplotstableread{Data/out_2003_11}{\vnine}
\addplot [Red, dotted, very thick] table[x=X, y=Z] {\vnine};
     
\pgfplotstableread{Data/out_2003_13}{\vnine}
\addplot [ForestGreen, dotted, very thick] table[x=X, y=Z] {\vnine};
    
\pgfplotstableread{Data/out_2003_15}{\vnine}
\addplot [Cyan, dotted, very thick] table[x=X, y=Z] {\vnine};

\end{semilogyaxis}
\end{tikzpicture}%
            }%
        }%
        \label{fig:bfout_2003}
    }
    \vspace{+0.20cm}
    \ref{named}
    \vspace{-0.10cm}
    \caption{\ac{DFR} of double-circulant codes ($r = 2003$), for several decoders and values of $v$. 
    Continuous and dotted lines refer to the standard and the TS-aware decoder, respectively.
    Each point has been estimated after (at least) $30$ decoding failures.}
    \label{fig:r2003}
\end{figure}


\begin{table}[h!]
\centering
\caption{$\bfout$ decoder thresholds used in Fig.~\ref{fig:bfout_2003}.}
\label{tab:thresholds_r2003}
\begin{tabular}{ccc}
\toprule
$v$ & & Thresholds\\\midrule
$9$ & & $(9, 8, 8, 8, 8, 7, 7, 7, 6, 6, 6, 6, 5, 5, 5)$\\
$11$ & & $(11, 10, 10, 9, 9, 8, 8, 8, 7, 7, 7, 7, 6, 6, 6)$\\
$13$ & & $(13, 12, 12, 11, 11, 10, 10, 9, 9, 8, 8, 8, 7, 7, 7)$\\
$15$ & & $(15, 14, 14, 13, 13, 12, 12, 11, 11, 10, 10, 9, 9, 8, 8)$\\\bottomrule
\end{tabular}
\end{table}

As it results from Fig.~\ref{fig:r2003}, $\bfmax$ is the decoder for which the benefits of the enhancement we propose are more evident. 
This is because, by design, this decoder flips only one bit in each iteration. 
So, whenever it encounters a residual error vector which consists in a $(u,u)$-almost near-codeword with $u > v/2$, it is very likely to converge to the close near-codeword. 
We also note that, in the waterfall region, the impact is relatively limited, since almost near-codewords are not expected to be the dominant source of decoding failures there.
A similar reasoning can be applied to the other two decoders: $\mld$ and $\bfout$. 
Indeed, also for them
the effect of the modification becomes more significant in the error-floor region.
Moreover, the fact that 
the reduction in terms of DFR appears less important is not surprising, since, for such two decoders, we do not control the number of flipped bits.
That is, there may occur iterations where no bits are flipped, as well as iterations where bits other than bad bits are flipped.
So, they do not gradually converge to a near-codeword, while $\bfmax$ is somewhat forced to do so.

\subsection{BIKE codes}

Let us we focus on the codes and decoders that are used in the BIKE cryptosystem.
We consider the NIST security category 1 parameters, that is, $r = 12323$, $v = 71$, $t = 134$, and the following three decoders: the already analyzed $\bfmax$ decoder with $2t$ iterations, the $\bgf$ decoder \cite{drucker2020qc} (which has been used in the first three rounds of the NIST competition), and the $\newbike$ decoder, described in version $5.2$ of the BIKE specification.
Modulo some minor differences regarding $\bgf$, both $\bgf$ and $\newbike$ are examples of a $\bfout$ decoder (decoding thresholds are chosen as a function of the iteration number and syndrome weight).

For these decoders, we simulate only decoding of $(t,u)$-almost near-codewords, for several values of $u$.
The results are shown in Fig.~\ref{fig:floor_bike}. 
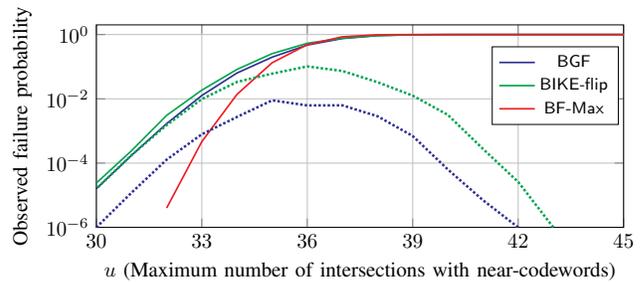
\begin{figure}[t]
    \centering
    \resizebox{\columnwidth}{!}{
\begin{tikzpicture} 
    \begin{semilogyaxis}[
        xlabel={$u$ (Maximum number of intersections with near-codewords)},
        ylabel={Observed failure probability},
        xmin=30, xmax=45,
        ymin=1e-6, ymax=2,
        xtick={30, 33, 36, 39, 42, 45},
        ytick={1, 1e-2, 1e-4, 1e-6, 1e-8},
        legend style={at={(0.75,0.7)},anchor=west, font=\footnotesize},
         ymajorgrids=true,
        yminorgrids=true,
        xmajorgrids=true,
        height  = 5.0cm,
        minor grid style = {gray!30!white},
        width = 1.2\columnwidth
    ]

\addplot [Blue, thick] coordinates{
(45, 1.000000e+00)
(44, 1.000000e+00)
(43, 1.000000e+00)
(42, 9.999980e-01)
(41, 9.999220e-01)
(40, 9.986562e-01)
(39, 9.865659e-01)
(38, 9.277264e-01)
(37, 7.531328e-01)
(36, 4.715155e-01)
(35, 1.974894e-01)
(34, 6.345852e-02)
(33, 1.268148e-02)
(32, 1.694530e-03)
(31, 1.740000e-04)
(30, 1.600000e-05)
(29, 0.000000e+00)

};
\addlegendentry{$\bgf$};

\addplot [Green, thick] coordinates{
(45, 1.000000e+00)
(44, 1.000000e+00)
(43, 1.000000e+00)
(42, 9.999720e-01)
(41, 9.996687e-01)
(40, 9.972301e-01)
(39, 9.836581e-01)
(38, 9.300411e-01)
(37, 7.706093e-01)
(36, 5.348259e-01)
(35, 2.545455e-01)
(34, 8.207705e-02)
(33, 1.860921e-02)
(32, 3.056001e-03)
(31, 2.427292e-04)
(30, 2.400000e-05)
(29, 0.000000e+00)
};
\addlegendentry{$\newbike$};

\addplot [Red, thick] coordinates{
(45, 1.000000e+00)
(44, 1.000000e+00)
(43, 1.000000e+00)
(42, 1.000000e+00)
(41, 1.000000e+00)
(40, 9.999960e-01)
(39, 9.994580e-01)
(38, 9.838940e-01)
(37, 8.515110e-01)
(36, 4.846410e-01)
(35, 1.324380e-01)
(34, 1.375400e-02)
(33, 4.640000e-04)
(32, 4.000000e-06)
};
\addlegendentry{$\bfmax$};
\addplot [Blue, very thick, densely dotted] coordinates{
(45, 0.000000e+00)
(44, 0.000000e+00)
(43, 0.000000e+00)
(42, 1.000000e-06)
(41, 7.000000e-06)
(40, 6.338843e-05)
(39, 6.960637e-04)
(38, 2.868000e-03)
(37, 6.265664e-03)
(36, 6.243172e-03)
(35, 8.966600e-03)
(34, 2.726467e-03)
(33, 7.901235e-04)
(32, 1.281308e-04)
(31, 1.200000e-05)
(30, 1.000000e-06)
(29, 0.000000e+00)
};

\addplot [Green, very thick, densely dotted] coordinates{
(45, 0.000000e+00)
(44, 0.000000e+00)
(43, 1.000000e-06)
(42, 2.600000e-05)
(41, 2.760753e-04)
(40, 3.165559e-03)
(39, 1.257071e-02)
(38, 3.292181e-02)
(37, 7.347670e-02)
(36, 1.019901e-01)
(35, 6.060606e-02)
(34, 3.350084e-02)
(33, 9.794319e-03)
(32, 1.528001e-03)
(31, 1.765303e-04)
(30, 1.600000e-05)
(29, 0.000000e+00)
};

\addplot [Red, very thick, densely dotted] coordinates{
(45, 0.000000e+00)
(44, 0.000000e+00)
(29, 0.000000e+00)
};

\end{semilogyaxis}
\end{tikzpicture}
    }
    \caption{Decoding failure probabilities on $(t,u)$-almost near-codewords for codes with $r = 12323$, $v = 71$ and $t = 134$. Continuous lines are referred to standard decoders, dotted lines to modified decoders. For each value of $u$, decoding of $10^8$ random almost near-codewords has been simulated.}
    \label{fig:floor_bike}
\end{figure}
Once again, $\bfmax$ seems to be the decoder with the lowest failure probability.
This is already true if we do not consider the improvement achievable with our approach, while, if we consider it, 
this becomes even more evident.
For all considered decoders, our approach allows reducing the failure probability significantly.
This is especially true for large values of $u$, for which decoders are more likely to converge to a near-codeword.
Remarkably,
for the modified $\bfmax$, we have never seen a decoding failure.
\vspace{-0.25cm}

\section{Conclusion}
\label{sec:concl}
We have introduced a technique to improve BF decoding of QC-MDPC codes leveraging knowledge of near-codewords, which can be applied to any BF decoder and comes with a very mild computational overhead.
With intensive numerical simulations, considering both toy codes and BIKE QC-MDPC codes, we have shown that the DFR can be reduced significantly, especially in the error floor region.
We have also shown that $\bfmax$, a recently proposed decoder, significantly outperforms other decoders, even those used by BIKE within the NIST competition.

\printbibliography

\end{document}